\documentclass[a4paper,12pt]{article}
\usepackage{amsmath, amssymb, amsthm, enumerate, bbm, graphicx, color,psfrag}
\newcommand{\e}{\mathrm{e}}
\newcommand{\Pb}{\mathbb{P}}

\newcommand{\I}{\mathrm{i}}
\newcommand{\dx}{\mathrm{d}}
\newcommand{\R}{\mathbb{R}}

\newcommand{\Z}{\mathbb{Z}}
\newcommand{\C}{\mathbb{C}}
\renewcommand{\Re}{\operatorname{Re}}
\renewcommand{\Im}{\operatorname{Im}}
\newcommand{\Tr}{\operatorname{Tr}}
\newcommand{\const}{\mathrm{const}}
\newcommand{\one}{\mathbbm{1}}
\newcommand{\eps}{\varepsilon}
\newcommand{\Or}{\mathcal{O}}
\newcommand{\diag}{\operatorname{diag}}
\newcommand{\GT}{\mathrm{GT}}
\newcommand{\lin}{\operatorname{span}}

\newtheorem{thm}{Theorem}
\newtheorem{prop}{Proposition}[section]
\newtheorem{lem}[prop]{Lemma}

\newtheorem{cla}[prop]{Claim}

\newtheorem{rem}[prop]{Remark}
\newenvironment{remark}{\begin{rem}\normalfont}{\end{rem}}

\title{Perturbed GUE Minor Process\\ and Warren's Process with Drifts}
\author{Patrik L.\ Ferrari\thanks{Institute for Applied Mathematics, Bonn University, Endenicher Allee 60, 53115 Bonn, Germany. E-mail: {\tt ferrari@uni-bonn.de}} \and
Ren\'e Frings\thanks{Institute for Applied Mathematics, Bonn University, Endenicher Allee 60, 53115 Bonn, Germany. E-mail: {\tt frings@uni-bonn.de}}
}

\date{}

\begin{document}
\sloppy
\maketitle

\begin{abstract}
We consider the minor process of (Hermitian) matrix diffusions with constant diagonal drifts. At any given time, this process is determinantal and we provide an explicit expression for its correlation kernel. This is a measure on the Gelfand-Tsetlin pattern that also appears in a generalization of Warren's process~\cite{War07}, in which Brownian motions have level-dependent drifts. Finally, we show that this process arises in a diffusion scaling limit from an interacting particle system in the anisotropic KPZ class in $2+1$ dimensions introduced in~\cite{BF08}. Our results generalize the known results for the zero drift situation.
\end{abstract}

\section{Introduction}
In this paper we determine a determinantal point process living on the Gelfand-Tsetlin cone $\GT_N$,
\begin{equation}
\GT_N = \{(x^1,x^2,\dotsc,x^N) \in \R^1 \times \R^2 \times \dotsb \times \R^N : x_k^{n+1} \leq x_k^n \leq x_{k+1}^{n+1} \}
\end{equation}
arising both from random matrices diffusions and interacting particle systems. An element $x \in \GT_N$ is called a Gelfand-Tsetlin pattern. Here is a graphical representation of $x \in \GT_4$, which illustrates the interlacing condition on $x$,
\begin{equation*}
\arraycolsep-.1ex
\begin{array}{ccccccccccccc}
x_1^4 & &  & & x_2^4 & &  & & x_3^4 & &  & & x_4^4 \\[-.5ex]
& \rotatebox[origin=c]{-45}{$\le$} & & \rotatebox[origin=c]{45}{$\le$} & & \rotatebox[origin=c]{-45}{$\le$} & & \rotatebox[origin=c]{45}{$\le$} & & \rotatebox[origin=c]{-45}{$\le$} & & \rotatebox[origin=c]{45}{$\le$}\\[-.5ex]
& & x_1^3 & & & & x_2^3 & & & & x_3^3 \\[-.5ex]
& & & \rotatebox[origin=c]{-45}{$\le$} & & \rotatebox[origin=c]{45}{$\le$} & & \rotatebox[origin=c]{-45}{$\le$} & & \rotatebox[origin=c]{45}{$\le$} \\[-.5ex]
& & & & x_1^2 & & & & x_2^2 \\[-.5ex]
& & & & & \rotatebox[origin=c]{-45}{$\le$} & & \rotatebox[origin=c]{45}{$\le$} \\[-.5ex]
& & & & & & x_1^1
\end{array}
\end{equation*}
Measures on Gelfand-Tsetlin patterns naturally appear in several fields of mathematics like (a) random matrix theory~\cite{Bar01,JN06,FN08,OR06} where the question of universality was recently approached in~\cite{Met11}, (b) random tiling problems~\cite{BG08,OR06,NY11}, (c) representation theory~\cite{BK07,BK09}, and (d) interacting particle systems~\cite{BF08,Nor08,NY11} and diffusions~\cite{War07,WW08}. Probably the most famous example which belongs to more than one of these classes is the Aztec diamond.
Indeed, the measure on the Aztec diamond both comes from a random tiling problem~\cite{EKLP92} and can be obtained through a Markov chain~\cite{Nor08} on a (discrete) Gelfand-Tsetlin pattern, which itself can be seen as a special case of the more general Markov chain construction in~\cite{BF08}. A continuous space analogue is for instance Warren's process, a system of interacting standard Brownian diffusions~\cite{War07,WW08}. Recently, the link between the ``drift-less'' case of~\cite{BF08} and Warren's process has been studied in~\cite{GS12}.

In this paper we consider GUE matrix diffusions with drifts. Its eigenvalue process at a fixed time is a determinantal point process and we explicitly determine its correlation kernel. We then show that the point process arises from a diffusion scaling limit of an interacting particle system in the anisotropic KPZ class in $2+1$ dimensions~\cite{BF07,BF08}, as well as in a generalization of Warren's process~\cite{War07} if we let the Brownian motions to have level-dependents drifts. The analogue results for the zero-drift case were all previously known, see~\cite{JN06,GS12,FF10}.

Remark that the GUE minor process and the Warren process are not the same if considered as stochastic processes. Indeed, this is true already for the zero-drift case. Without drift, the GUE minor process was described in~\cite{JN06}, while Warren's process was introduced in~\cite{War07}. It is known that the two processes coincide when projected on ``space-like'' paths~\cite{FF10}, in which case they both are Markovian and determinantal. However, in the whole ``space-time'' the two processes are different~\cite{ANvM10b}. Here we focus on the fixed-time process although the connection will certainly hold along ``space-like'' paths as for the zero-drift case, see~\cite{FF10}.

\subsubsection*{Matrix diffusions}
The first model we study on $\GT_N$ is a variant of the GUE minor process which has been introduced in \cite{JN06}. Consider an $N \times N$ Hermitian matrix $H$ with eigenvalues $\lambda_1^N\leq \dotsb \leq \lambda_N^N$. Denote by $H^n$ the submatrix obtained by keeping the first $n$ rows and columns of $H$, and its ordered eigenvalues by $\lambda^n_1\leq \dotsb \leq \lambda_n^n$. The collection of all these eigenvalues $(\lambda^1,\dotsc,\lambda^N)$ then forms a Gelfand-Tsetlin pattern, with $\lambda^n=(\lambda_1^n,\dotsc,\lambda_n^n)$. In this paper we take $H(t)$ to be a GUE matrix diffusion perturbed by a deterministic drift matrix $M=\diag(\mu_1,\dotsc,\mu_N)$, i.e., we consider $G(t)=H(t)+t M$ with $H$ evolving as standard \mbox{GUE Dyson's} Brownian Motion starting from $0$. The eigenvalues' point process $\xi$ has support on $\R\times \{1,\dotsc,N\}$,
\begin{equation}
\xi(\dx x,m)=\sum_{1\leq k \leq n \leq N}\delta_{n,m} \delta_{\lambda_k^n}(\dx x)
\end{equation}
and its correlation function is given as follows (see Section~\ref{sec:2} for the proof).

\begin{thm}\label{thm:1}
For a fixed time $t>0$ consider the eigenvalues' point process on the $N$ submatrices of $H(t)$. Then, its $m$-point correlation function $\varrho^m_t$ is given by
\begin{equation}\label{eq:24}
\varrho^m_t((x_1,n_1),\dots,(x_m,n_m))=\det[K_t((x_i,n_i),(x_j,n_j))]_{1\leq i,j\leq m},
\end{equation}
with $(x_i,n_j) \in \R \times \{1,\dots,N\}$ and correlation kernel
\begin{equation}\label{eq:26}
K_t((x,n),(x',n'))=-\phi^{(n,n')}(x,x')+\sum_{k=1}^{n'} \Psi_{n-k}^{n,t}(x) \Phi_{n'-k}^{n',t}(x'),
\end{equation}
where\footnote{ For a set $S$, the notation $\frac{1}{2\pi\I} \oint_{\Gamma_S}\dx w\,f(w)$ means that the integral is taken over any positively oriented simple contour that encloses only the poles of $f$ belonging to $S$.}
\begin{align}
\phi^{(n,n')}(x,x') &= \frac{(-1)^{n'-n}}{2\pi\I} \int_{\I\R+\mu_-}\dx z\, \frac{\e^{z(x'-x)}}{(z-\mu_{n+1})\dotsb(z-\mu_{n'})}\,\one_{[n<n']}, \label{eq:9}\\
\Psi_{n-k}^{n,t}(x) & = \frac{(-1)^{n-k}}{2\pi \I} \int_{\I\R+\mu_-} \dx z \, \e^{tz^2/2-xz} \, \frac{(z-\mu_1)\dotsb (z-\mu_n)}{(z-\mu_1)\dotsb(z-\mu_k)}, \label{eq:9a} \\
\Phi_{n-\ell}^{n,t}(x) &= \frac{(-1)^{n-\ell}}{2\pi \I} \oint_{\Gamma_{\mu_1,\dotsc,\mu_N}} \dx w\, \e^{-tw^2/2+x w} \, \frac{(w-\mu_1)\dotsb(w-\mu_{\ell-1})}{(w-\mu_1)\dotsb(w-\mu_n)} \label{eq:9b}
\end{align}
with $\mu_-<\min\{\mu_1,\dots,\mu_N\}$.
\end{thm}
\begin{rem}
The integral for $\phi^{(n,n')}$ in \eqref{eq:9} is only well-defined for $n'-n>1$. For $n'-n=1$ we set $\phi^{(n-1,n)}(x,x'):=\phi_n(x,x')=\e^{\mu_n(x'-x)}\one_{[x>x']}$ instead.
\end{rem}

In an independent work~\cite{AvMW13} on minors of random matrices by Adler, van Moerbeke, and Wang appeared on the arXiv after this work, the same kernel is computed and a double integral expression is also provided.

\subsubsection*{Warren's process with drifts}
Our second model is Warren's process with drifts that describes the dynamics of a system of Brownian motions $\{B_k^n,1\leq k \leq n \leq N\}$ on $\GT_N$, where $B_1^1$ is a standard Brownian motion with drift $\mu_1$ starting from the origin. The Brownian motions $B_1^2$ and $B_2^2$ are Brownian motions with drifts $\mu_2$ conditioned to start at the origin and, whenever they touch $B_1^1$, they are reflected off $B_1^1$. Similarly for $n\geq 2$, $B_k^n$ is a Brownian motion with drift $\mu_n$ conditioned to start at the origin and being reflected off $B_k^{n-1}$ (for $k\leq n-1$) and $B_{k-1}^{n-1}$ (for $k\geq 2$). The process with $\mu_1=\dotsb=\mu_N=0$ was introduced and studied by Warren in~\cite{War07}.

The correlation functions of this process at a fixed time agree with those of the perturbed GUE minor process (the proof is in Section~\ref{sec:3}).

\begin{thm}\label{thm:2}
For a fixed time $t>0$ consider the point process of the positions of the Brownian motions $\{B_k^n(t):1\leq k\leq n\leq N\}$ described above. Then, its $m$-point correlation function $\varrho^m_t$ is also given by \eqref{eq:24}.
\end{thm}

%

\subsubsection*{Interacting particle system}
Finally we introduce a discrete model giving rise to Warren's process with drifts under a diffusion scaling limit. This model is a generalization of TASEP with particle-dependent jump rates~\cite{BF07} to the $2+1$ dimensional particle system with Markov dynamics introduced in~\cite{BF08}. We denote by $x_k^n \in \Z$ the position of a particle labeled by $(k,n)$, with $1\leq k \leq n \leq N$, and call $n$ the ``level'' of the particle. Particle $(k,n)$ performs a continuous time random walk with one-sided jumps (to the right) and with rate $v_n$. Particles with smaller level evolve independently from the ones with higher level like in the Brownian motion model described above. More precisely, the interaction between levels is the following: (a) if particle $(k,n)$ tries to jump to $x$ and $x_{k-1}^{n-1}=x$, then the jump is suppressed, and (b) when particle $(k,n)$ jumps from $x-1$ to $x$, then all particles labeled by $(k+\ell,n+\ell)$ (for some $\ell\geq 1$) which were at $x-1$ are forced to
jump to $x$, too. This is a particle system with state space in a discrete Gelfand-Tsetlin pattern.

Consider the diffusion scaling with appropriate scaled jump rates
\begin{equation}\label{scaling}
t=\tau T,\quad x_k^n=\tau T-\sqrt{T} \lambda_k^n, \quad v_n=1-\frac{\mu_n}{\sqrt{T}}.
\end{equation}
Then, in the $T\to\infty$ limit, the particle process $\{x_k^n(t)\}$ converges to the GUE minor process with drift $\{\lambda_k^n(\tau)\}$.

More precisely, let us denote by $\widetilde \Pb^v$ the probability measure on these particles with jump rates $v=(v_1,\dotsc,v_N)$ given in \eqref{scaling}. We fix $\tau>0$ and set
\begin{equation}
\nu_T(A) = \widetilde \Pb^{v}\biggl(-\frac{x_k^n(\tau T)-\tau T}{\sqrt T} \in A_k^n \text{ for all } 1\leq k\leq n\leq N\biggr)
\end{equation}
where $A_k^n \subseteq \R$ are Borel sets, $A=\prod_{1\leq k\leq n \leq N}A_k^n$. Moreover, we define
\begin{equation}
\nu(A) = \Pb^\mu\bigl(\lambda_k^n(\tau)\in A_k^n \text{ for all } 1\leq k\leq n\leq N \bigr)
\end{equation}
where $\Pb^\mu$ is the GUE minor measure with drift $\diag(\mu_1,\dotsc,\mu_N)$. In Section~\ref{sec:4} we show the following result.
\begin{thm}\label{thm:3}
 As $T\to \infty$, $\nu_T$ converges to $\nu$ in total variation, i.e.,
\begin{equation}\label{eq:31}
\lim_{T\to \infty} \sup_{\substack{A \subseteq \R^{N(N+1)/2},\\ A \, \mathrm{ Borel }}} \lvert \nu_T(A) - \nu(A)\rvert = 0.
\end{equation}
In particular, $\nu_T \to \nu$ weakly.
\end{thm}

\subsubsection*{Acknowledgments}
The authors are grateful to M. Shkolnikov to point out how to use~\cite{GS12} to show Theorem~\ref{thm:2} and to J. Warren for useful discussions. This work was supported by the German Research Foundation via the SFB~611--A12 and SFB 1060--B04 projects.

\newpage
\section{GUE minor process with drift}\label{sec:2}

\subsection{Model and measure}

Let $(H(t):t\geq 0)$ be a process on the $N\times N$ Hermitian matrices defined by
\begin{equation}
H_{k\ell}(t) =
\begin{cases}
b_{kk}(t)+\mu_k t, & \text{if } 1\leq k \leq N, \\
\frac1{\sqrt{2}} (b_{k\ell}(t)+\I \tilde b_{k\ell}(t)), & \text{if } 1\leq k < \ell \leq N, \\
\frac1{\sqrt{2}} (b_{k\ell}(t)-\I \tilde b_{k\ell}(t)), & \text{if } 1\leq \ell < k \leq N,
\end{cases}
\end{equation}
where $\{b_{kk}, b_{k\ell}, \tilde b_{k\ell}\}$ are independent one-dimensional standard Brownian motions\footnote{Here, standard Brownian motions start from $0$ and are normalized to have variance $t$ at time $t$.}. Denote by $M=\diag(\mu_1,\dots,\mu_N)$ the diagonal drift matrix added to the matrix $H$. Then, the probability measure on these matrices at time $t$ is given by
\begin{equation}\label{eq:2}
\Pb(H\in\dx H)=\const\times \exp\left(-\frac{\Tr(H-t M)^2}{2t}\right)\dx H
\end{equation}
where $\dx H=\prod_{i=1}^N \dx H_{ii}\prod_{1\leq j < k \leq N} \dx\Re(H_{j,k}) \dx\Im(H_{j,k})$ and $\const$ is the normalization constant.

Since we are interested in the statistics of the eigenvalues' minors at time $t$, we first determine the measure on the eigenvalues of the $N\times N$ matrix.
\begin{lem}\label{lem:3}
Assume that $\mu_1,\dotsc,\mu_N$ are all distinct. Then under \eqref{eq:2}, the joint probability measure of the eigenvalues $\lambda_1,\dots,\lambda_N$ of $H$ is given by
\begin{multline}\label{eq:8}
\Pb(\lambda_1 \in \dx \lambda_1, \dots \lambda_N \in \dx\lambda_N)\\  = \const\times \det\bigl[ \e^{-(\lambda_i-t \mu_j)^2/(2t)}\bigr]_{1\le i,j\le N} \, \frac{\Delta(\lambda_1,\dots,\lambda_N)}{\Delta(\mu_1,\dots,\mu_N)} \,\dx \lambda_1\dotsb\dx \lambda_N
\end{multline}
with $\const$ a normalization constant and $\Delta(x_1,\dots,x_m) = \prod_{1\leq i < j \leq m}(x_j-x_i)$ the Vandermonde determinant.
\end{lem}

\begin{remark}
If $\mu_1,\dotsc,\mu_N$ are not all distinct, we have to take limits in \eqref{eq:8}. For instance, if $\mu_1=\dotsb=\mu_N\equiv \mu$, then
\begin{multline}
\Pb(\lambda_1 \in \dx \lambda_1, \dots \lambda_N \in \dx\lambda_N)\\  = \const\times \biggl(\prod_{k=1}^N \e^{-(\lambda_k-t \mu)^2/(2t)}\biggr) \Delta^2(\lambda_1,\dots,\lambda_N) \,\dx \lambda_1\dotsb\dx \lambda_N.
\end{multline}
\end{remark}

\begin{proof}[Proof of Lemma~\ref{lem:3}]
We diagonalize $H=U\Lambda U^\ast$ with a unitary matrix $U$ and the diagonal matrix $\Lambda=\diag(\lambda_1,\dotsc,\lambda_N)$. Then,
\begin{equation}\label{eq:6}
\e^{-\Tr(H-t M)^2/(2t)} \dx H = \const \times \e^{-\Tr(U \Lambda U^\ast -tM)^2/(2t)}\Delta^2(\lambda)\, \dx U\,\dx \lambda,
\end{equation}
where $\dx U$ is the Haar measure on the unitary group $\mathcal U$. Moreover, since
\begin{equation}
\Tr (U\Lambda U^\ast - tM)^2 = \Tr \Lambda^2+t^2 \Tr M^2-2t\Tr(U\Lambda U^\ast M)
\end{equation}
by integrating over $\mathcal U$ in \eqref{eq:6} and using the Harish-Chandra-Itzykson-Zuber formula, we obtain the desired expression.
\end{proof}

Now we focus on the minor process. For $1\leq n\leq N$ let us denote by $H^n(t)$ the $n\times n$ principal submatrix of $H(t)$ which is obtained from $H(t)$ by keeping only the $n$ first rows and columns. In particular, $H^1(t)=H_{11}(t)$ and $H^N(t)=H(t)$. We denote by $\lambda_1^n(t)\leq \dots \leq \lambda_n^n(t)$ the ordered eigenvalues of $H^n(t)$. It is then a classical fact of linear algebra that at any time $t$, the process $(\lambda^1,\dots,\lambda^N)(t)$ lies in the Gelfand-Tsetlin cone of order $N$,
\begin{equation}
\GT_N = \{ (x^1,\dots,x^N) \in \R^1 \times \dotsb \times \R^N : x^n\preceq x^{n+1} \text{ for all } 1\leq k \leq N-1\},
\end{equation}
where $x^n \preceq x^{n+1}$ means that $x^n$ and $x^{n+1}$ interlace, i.e.,
\begin{equation}
x_k^{n+1} \leq x^n_k \leq x^{n+1}_{k+1} \quad \text{ for all } 1 \leq k \leq n.
\end{equation}

The induced measure on $\{\lambda_k^n:1\leq k \leq n \leq N\}$ is the following.
\begin{prop}
Fix $t>0$. Then, under the measure \eqref{eq:2}, the joint density of the eigenvalues of $\{H^n:1\leq n \leq N\}$ on $\GT_N$ is given by
\begin{equation}\label{eq:1}
\const\times \prod_{k=1}^N \e^{-t\mu_k^2/2}\prod_{k=1}^N \e^{-(\lambda_k^N)^2/(2t)} \Delta(\lambda^N)\prod_{\substack{1\leq n \leq N \\ 1 \leq k \leq n}} \e^{\mu_n \lambda_k^n} \prod_{\substack{2\leq n\leq N\\1\leq k \leq n-1}} \e^{-\mu_n\lambda_k^{n-1}},
\end{equation}
where the normalization constant does not depend on $\mu_1,\dots,\mu_N$.
\end{prop}

\begin{proof}
We first derive \eqref{eq:1} under the assumption that the $\mu_1,\dots,\mu_N$ are all distinct; the case where some of the $\mu_i$ are equal is then recovered by taking the limit. We prove the statement inductively and follow the presentation in \cite{FN08b}. For $N=1$, the density is clearly proportial to $\exp(-(\lambda_1^1-\mu_1t)^2/(2t))$. For $N\geq 2$, we consider an $N \times N$ matrix $H^N$ distributed according to \eqref{eq:2} which we write as
\begin{equation}
H^N - t M = \begin{pmatrix} H^{N-1} & w \\ w^\ast & x \end{pmatrix} - t \begin{pmatrix} M^{N-1} & 0 \\ 0 & \mu_N \end{pmatrix},
\end{equation}
where $M^{N-1}$ denotes the $(N-1)\times (N-1)$ principal submatrix of $M$, $w \in \C^{N-1}$ is a Gaussian vector and $x \in \R$ is a Gaussian variable. Then we diagonalize $H^{N-1}$, i.e., we choose a unitary matrix $U$ such that $H^{N-1}=U \Lambda U^\ast$ with $\Lambda = \diag(\lambda^{N-1}_1,\dotsc,\lambda^{N-1}_{N-1})$ the diagonal matrix for the eigenvalues. Since the Gaussian distribution is invariant under unitary rotations and $w$ is independent of $H^{N-1}$, we have
\begin{equation}
\begin{pmatrix} U^\ast & 0 \\ 0 & 1 \end{pmatrix} (H^N-tM) \begin{pmatrix} U & 0 \\ 0 & 1 \end{pmatrix} \overset{d}{=} \begin{pmatrix} \Lambda & w \\ w^\ast & x \end{pmatrix}-t\begin{pmatrix} U^\ast M^{N-1}U & 0 \\ 0 & \mu_N  \end{pmatrix},
\end{equation}
where $\overset{d}{=}$ denotes equality in distribution. Applying the map $H^{N-1} \mapsto (\Lambda,U)$, we get that measure \eqref{eq:2} on $H^N$ is proportional to
\begin{multline}\label{eq:3}
\exp\left(-\frac1{2t}\Tr \left[\begin{pmatrix} \Lambda & w \\ w^\ast & x \end{pmatrix}-t\begin{pmatrix} U^\ast M^{N-1}U & 0 \\ 0 & \mu_N  \end{pmatrix}\right]^2\right) \Delta^2(\lambda^{N-1}) \\
\times \dx U \, \dx w \, \dx x \, \dx\lambda^{N-1},
\end{multline}
where $\dx U$ is the Haar measure on the unitary group $\mathcal U_{N-1}$. We consider only the part of \eqref{eq:3} that depends on $U$ and integrate over $\mathcal U_{N-1}$, using the Harish-Chandra-Itzykson-Zuber formula,
\begin{equation}
\int_{\mathcal U_{N-1}} \dx U \, \e^{\Tr (\Lambda U^\ast M^{N-1} U)} = \const \times \frac{\det[\e^{\lambda_i^{N-1}\mu_j}]_{1\leq i,j\leq N-1}}{\Delta(\lambda^{N-1})\Delta(\mu_1,\dotsc,\mu_{N-1})}.
\end{equation}
After this integration the measure \eqref{eq:3} reads
\begin{equation}\label{eq:4}
\const \times \Pb(\lambda^{N-1}\in\dx \lambda^{N-1}) \, \e^{x \mu_N-t\mu_N^2/2} \prod_{k=1}^{N-1} \e^{-\lvert w_k \rvert^2/t} \dx x\,\dx w.
\end{equation}
We focus on the measure on $w_k$ and represent the variables in polar coordinates, $w_k = r_k\e^{\I\varphi_k}$ with $r_k \in \R_+$ and $\varphi_k \in [0,2\pi)$. Since the Jacobian of this transformation is given by $r_1\dotsb r_{N-1}$, we get
\begin{equation}
\prod_{k=1}^{N-1} \e^{-\lvert w_k\rvert^2/t}\dx w_k = \prod_{k=1}^{N-1} r_k \e^{-r_k^2/t}\dx r_k \,\dx \varphi_k,
\end{equation}
where $\dx r_k$ and $\dx \varphi_k$ are Lebesgue measures on $\R_+$ and $[0,2\pi)$. Then we can express $r_k$ and $x$ in terms of the eigenvalues of $H^{N-1}$ and $H^N$, see e.g.\ \cite{FN08b} for details,
\begin{equation}
\begin{aligned}
r_k^2 &=  - \frac{\prod_{j=1}^N (\lambda_k^{N-1}-\lambda_j^N)}{\prod_{j=1,j\neq k}^{N-1}(\lambda_k^{N-1}-\lambda_j^{N-1})}\, \one_{[\lambda^{N-1} \preceq \lambda^N]},\\
x & = \Tr(H^N-H^{N-1})=\sum_{i=1}^N \lambda_k^N-\sum_{k=1}^{N-1} \lambda_k^{N-1}.
\end{aligned}
\end{equation}
The Jacobian of the transformation $T:(r_1,\dotsc,r_{N-1},x) \mapsto \lambda^N$ is then given by
\begin{equation}
r_1\dotsb r_{N-1}\lvert \det T' \rvert = \frac{\Delta(\lambda^N)}{\Delta(\lambda^{N-1})} \, \one_{[\lambda^{N-1} \preceq \lambda^N]},
\end{equation}
and hence, given $\lambda^{N-1}$, we have
\begin{multline}\label{eq:5}
\e^{x \mu_N} \prod_{k=1}^{N-1} \e^{-\lvert w_k \rvert^2/t} \dx x\,\dx w
=  \prod_{k=1}^N \e^{-(\lambda_k^N)^2/(2t)+\mu_N \lambda_k^N} \prod_{k=1}^{N-1} \e^{(\lambda_k^{N-1})^2/(2t)-\mu_N \lambda_k^{N-1}}\\
\times \frac{\Delta(\lambda^N)}{\Delta(\lambda^{N-1})} \, \one_{[\lambda^{N-1}\preceq \lambda^N]} \, \dx\lambda^N \, \dx \varphi.
\end{multline}
Here we used that $2 (r_1^2+\dotsb+r_{N-1}^2) = \Tr (H^N)^2-\Tr (H^{N-1})^2$. Moreover, by the induction assumption for $N-1$ we have
\begin{multline}\label{eq:30}
\Pb(\lambda^{N-1}\in\dx\lambda^{N-1}) = \const\times \prod_{k=1}^{N-1}\e^{-t \mu_k^2/2} \prod_{k=1}^{N-1} \e^{-(\lambda^{N-1}_k)^2/(2t)}\Delta(\lambda^{N-1}) \\ \times \prod_{\substack{1\leq n\leq N-1 \\ 1\leq k\leq n}} \e^{\mu_n\lambda_k^n}\prod_{\substack{2\leq n\leq N-1\\1\leq k\leq n}}\e^{-\mu_n\lambda_k^{n-1}} \prod_{n=1}^{N-1}\dx \lambda^n.
\end{multline}
Finally, inserting \eqref{eq:5} and \eqref{eq:30} into \eqref{eq:4} and integrating out $\varphi$ (which multiplies the measure by a finite constant) results in the claimed formula \eqref{eq:1}.
\end{proof}

\subsection{Correlation functions}

Now we determine the correlation functions of the point process on the eigenvalues $\{\lambda_k^n:1\leq k\leq n\leq N\}$, and for that purpose, we rewrite the density in \eqref{eq:1} as a product of determinants. We set $\phi_n(x,y)=\e^{\mu_n(y-x)}\one_{\{x>y\}}$ and introduce ``virtual'' variables $\lambda_n^{n-1}=\text{virt}$ with the property that $\phi_n(\text{virt},y)=\e^{\mu_n y}$. Then in \eqref{eq:1} we have, up to a set of measure zero,
\begin{equation}
\det[\phi_n(\lambda_i^{n-1},\lambda_j^n)]_{1\leq i,j\leq n} = \prod_{j=1}^n \e^{\mu_n \lambda_j^n} \prod_{j=1}^{n-1} \e^{-\mu_n \lambda_j^{n-1}} \one_{[\lambda^n \preceq \lambda^{n+1}]}.
\end{equation}
Moreover, for $k=1,\dots,N$ we set
\begin{equation}
\Psi_{N-k}^{N,t}(x)=\frac{\e^{-x^2/(2t)}}{\sqrt{2\pi t}}\,t^{-(N-k)/2}\,p_{N-k}\left(\frac{\mu_{k+1}t-x}{\sqrt t},\dots,\frac{\mu_Nt-x}{\sqrt t}\right),
\end{equation}
where $p_n$  are symmetric polynomials of degree $n$ in $n$ variables defined by $p_0\equiv 1$ and
\begin{equation}
p_n(x_1,\dots,x_n)=\frac{(-1)^n}{\I \sqrt{2\pi}}\int_{\I \R}\dx w\, \e^{w^2/2}(w-x_1)\dotsb(w-x_n) \quad \text{ for } n\geq 1.
\end{equation}
Hence we have that
\begin{equation}
\prod_{k=1}^N \e^{-(\lambda_k^N)^2/(2t)} \Delta(\lambda^N) = \const \times \det\bigl[\Psi^{N,t}_{N-k}(\lambda_\ell^N)\bigr]_{1\leq k,\ell\leq N},
\end{equation}
which means that we can rewrite \eqref{eq:1} as
\begin{equation}\label{eq:10}
\const \times \prod_{n=1}^N \det \bigl[\phi_n (\lambda_i^{n-1},\lambda_j^n)\bigr]_{1\leq i,j\leq n} \prod_{k=1}^N \e^{-t\mu_k^2/2} \det\bigl[\Psi^{N,t}_{N-k}(\lambda_\ell^N)\bigr]_{1\leq k,\ell\leq N} .
\end{equation}
Note that by a change of variable $w = (t z-x)/\sqrt t$ we have
\begin{equation} \label{eq:11}
\Psi_{N-k}^{N,t}(x)=\frac{(-1)^{N-k}}{2\pi\I} \int_{\I\R} \dx z\, \e^{tz^2/2-xz} (z-\mu_{k+1})\dotsb(z-\mu_N).
\end{equation}
A measure of the form \eqref{eq:10} has determinantal correlation functions and the kernel can be computed with Lemma~3.4 of \cite{BFPS06}, see the appendix.
\subsection{Proof of Theorem~\ref{thm:1}}

We prove the theorem first for $\mu_1<\dots<\mu_N$ and then use analytic continuation. Note that for $n=N$, the function $\Psi_{n-k}^{n,t}$ in \eqref{eq:9a} is the same as $\Psi_{N-k}^{N,t}$ in \eqref{eq:11}.

\begin{lem}\label{lem:1} The following identities hold.
\begin{itemize}
\item[(i)] For all $n\in\{1,\dots,N\}$, $k\in\Z$ and $t>0$, we have $\phi_n \ast \Psi^{n,t}_{n-k} = \Psi^{n-1,t}_{n-1-k}$.
\item[(ii)] For $n<n'$, we have $\phi_{n+1}\ast\dots\ast\phi_{n'} = \phi^{(n,n')}$ with $\phi^{(n,n')}$ given in \eqref{eq:9}.
\end{itemize}
\end{lem}

\begin{proof}
Because of $\Re z < \mu_n$ we can exchange the two integrals,
\begin{equation}
\begin{aligned}
(\phi_n &\ast \Psi^{n,t}_{n-k})(x) \\
& = \int_{-\infty}^x\dx y\, \e^{\mu_n(y-x)} \frac{(-1)^{n-k}}{2\pi \I} \int_{\I \R+\mu_-} \hspace{-1em} \dx z \, \e^{tz^2/2-yz} \, \frac{(z-\mu_1)\dots (z-\mu_n)}{(z-\mu_1)\dots(z-\mu_k)} \\
& = \frac{(-1)^{n-k}}{2\pi \I} \int_{\I \R+\mu_-} \hspace{-1em} \dx z\, \e^{tz^2/2-\mu_nx} \, \frac{(z-\mu_1)\dots (z-\mu_n)}{(z-\mu_1)\dots(z-\mu_k)} \int_{-\infty}^x\dx y\, \e^{y(\mu_n-z)} \\
& = \frac{(-1)^{n-1-k}}{2\pi \I} \int_{\I\R+\mu_-}\dx z\, \e^{tz^2/2-x z} \, \frac{(z-\mu_1)\dots (z-\mu_{n-1})}{(z-\mu_1)\dots (z-\mu_k)} \\
& = \Psi^{n-1,t}_{n-1-k}(x).
\end{aligned}
\end{equation}
This proves the first statement. To show (ii), we first consider the case $n'-n=2$. A simple calculation gives
\begin{equation}
\phi^{(n-2,n)}(x,x')=(\phi_{n-1} \ast \phi_n)(x,x')= -\left( \frac{\e^{\mu_n(x'-x)}}{\mu_n-\mu_{n-1}}+\frac{\e^{\mu_{n-1}(x'-x)}}{\mu_{n-1}-\mu_n}\right)\one_{[x>x']},
\end{equation}
which has the following contour integral representation,
\begin{equation}
\phi^{(n-2,n)}(x,x') = \frac1{2\pi\I} \int_{\I \R+\mu_-} \dx z\,\frac{\e^{z(x'-x)}}{(z-\mu_{n-1})(z-\mu_n)}.
\end{equation}
For $n'-n>2$, we get inductively that
\begin{equation}
\begin{aligned}
(\phi_n &\ast \phi^{(n,n')})(x,x') \\
& = \frac{(-1)^{n'-n}}{2\pi\I} \int_{-\infty}^x\dx y \, \e^{\mu_n(y-x)} \int_{\I\R+\mu_-} \hspace{-1em} \dx z\, \frac{\e^{z(x'-y)}}{(z-\mu_{n+1})\dotsb(z-\mu_{n'})} \\
& = \frac{(-1)^{n'-n}}{2\pi\I}  \int_{\I\R+\mu_-} \hspace{-1em} \dx z\, \frac{\e^{zx'-\mu_nx}}{(z-\mu_{n+1})\dotsb(z-\mu_{n'})} \int_{-\infty}^x \dx y\, \e^{y(\mu_n-z)} \\
& = \frac{(-1)^{n'-n+1}}{2\pi\I} \int_{\I\R+\mu_-}\dx z\, \frac{\e^{z(x'-x)}}{(z-\mu_n)(z-\mu_{n+1})\dotsb(z-\mu_{n'})} \\
& = \phi^{(n-1,n')}(x,x'),
\end{aligned}
\end{equation}
where, as before, we could exchange the integrals because of $\Re z< \mu_n$.
\end{proof}

Next we consider the $n$-dimensional space $V_n$ spanned by the set of functions
\begin{equation}
\{\phi_1 \ast \phi^{(1,n)}(x_1^0,\cdot),\dots,\phi_{n-1}\ast \phi^{(n-1,n)}(x_{n-1}^{n-2},\cdot),\phi_n(x_n^{n-1},\cdot)\}.
\end{equation}
According to Lemma~3.4 of \cite{BFPS06} we need to find a basis $\{\Phi_{n-k}^{n,t}:1\leq k\leq n\}$ of $V_n$
that is biorthogonal to the set $\{\Psi_{n-k}^{n,t}:1\leq k\leq n\}$, i.e.,
\begin{equation}
\int_\R \dx x \, \Psi_{n-k}^{n,t}(x)\Phi_{n-\ell}^{n,t}(x) = \delta_{k\ell}, \quad 1 \leq k,\ell\leq n.
\end{equation}
The form of the biorthogonal functions can be guessed, with some experience, from the form of the kernel~\cite{BP07}.

\begin{lem}\label{lem:1b} We have:
\begin{itemize}
\item[(i)] $V_n$ is spanned by $\{x \mapsto \e^{\mu_k x}: 1\leq k \leq n\}$.
\item[(ii)] The functions $\{\Phi_{n-k}^{n,t}:1\leq k\leq n\}$ are given by \eqref{eq:9b}.
\end{itemize}
\end{lem}

\begin{proof}
For any $\eps>0$ we have
\begin{multline}
(\phi_k \ast \phi^{(k,n)})(x_k^{k-1},x) \\
 = \int_\R \dx y\, \e^{\mu_k y} \frac{(-1)^{n-k}}{2\pi\I}\int_{\I \R+\mu_{k+1}-\eps} \dx z\,\frac{\e^{z(x-y)}}{(z-\mu_{k+1})\dotsb(z-\mu_n)}.
\end{multline}
We split the $y$-integral into one over $\R_+$ and one over $\R_-$. Then we can exchange the integrals over $\R_-$ and the imaginary axis provided that \mbox{$\Re z < \mu_k$} and use $\int_{\R_-} \dx y \, \e^{y(\mu_k-z)}=\frac{1}{z-\mu_k}$. In the same way we integrate over $\R_+$ taking $z$ such that $\mu_k<\Re z < \mu_{k+1}$. This gives $\int_{\R_+}\dx y\,\e^{y(\mu_k-z)}=-\frac1{z-\mu_k}$. Putting these two integrals together we get
\begin{equation}
(\phi_k \ast \phi^{(k,n)})(x_k^{k-1+1},x) = \frac{(-1)^{n-k}}{2\pi\I} \oint_{\Gamma_{\mu_k}}\dx z\, \frac{\e^{x z}}{(z-\mu_k)(z-\mu_{k+1})\dotsb(z-\mu_n)},
\end{equation}
which is a constant multiple of $\e^{\mu_k x}$. This proves (i). For (ii) we proceed similarly. Using that $1\leq k,\ell\leq n$ we have
\begin{multline}\label{eq:12}
\int_\R \dx x \, \Psi_{n-k}^{n,t}(x)\Phi_{n-\ell}^{n,t}(x) \\
= \frac{(-1)^{k+\ell}}{(2\pi\I)^2} \int_\R\dx x\int_{\I \R} \dx z \oint_{\Gamma_{\mu_\ell,\dots,\mu_n}}\dx w\, \frac{\e^{tz^2-x z}}{\e^{tw^2-x w}}\,\frac{(z-\mu_{k+1})\dotsb(z-\mu_n)}{(w-\mu_\ell)\dotsb(w-\mu_n)}.
\end{multline}
When integrating $x$ over $\R_-$, we take the $z$-integral such that $\Re z<\Re w$, and when we integrate $x$ over $\R_+$, we choose $\Re z> \Re w$. Thus, \eqref{eq:12} reduces to
\begin{equation}
\frac{(-1)^{k+\ell}}{2\pi\I} \oint_{\Gamma_{\mu_\ell,\dots,\mu_n}} \dx w\, \frac{(w-\mu_{k+1})\dotsb(w-\mu_n)}{(w-\mu_\ell)\dotsb(w-\mu_n)} = \delta_{k\ell}.
\end{equation}
Finally, note that
\begin{equation}
\Phi_{n-\ell}^{n,t}(x)= \sum_{i=\ell}^n b_i \e^{\mu_i x}\quad \text{ with} \quad b_i= \prod_{\substack{j=\ell\\j\neq i}}^n \frac{\e^{-t \mu_i^2/2}}{\mu_i-\mu_j}.
\end{equation}
which shows that the set $\{\Phi_{n-k}^{n,t}:1\leq k\leq n\}$ spans $V_n$.
\end{proof}
Next we verify Assumption (A) from Lemma~3.4 in \cite{BFPS06}. Indeed,
\begin{equation}
\Phi_0^{n,t}(x) = \frac1{2\pi\I}\oint_{\Gamma_{\mu_n}} \dx w\, \frac{\e^{-tw^2/2+xw}}{w-\mu_n} = c_n \phi_n(x_n^{n-1},x)
\end{equation}
with $c_n=\e^{-t\mu_n^2/2}\neq 0$ for $n=1,\dots,N$.

Finally, we can also determine the value of the normalization constant in \eqref{eq:10}, since it is given by $1/\det [M_{k\ell}]_{1\leq k,\ell\leq N}$ with
\begin{equation}
M_{k\ell}=(\phi_k \ast \dotsb \ast \phi_N \ast \Psi_{N-\ell}^{N,t})(x_k^{k-1}).
\end{equation}

\begin{lem}\label{lem:2}
We have $\det M = \prod_{n=1}^N \e^{t\mu_n^2/2}$, in particular $\det M >0$.
\end{lem}

\begin{proof} By Lemma~\ref{lem:1} (i)  we may write $M_{k\ell} = (\phi_k \ast \Psi_{k-\ell}^{k,t})(x_k^{k-1})$. Thus, for $k\geq \ell$,
\begin{equation}
M_{k\ell} = \int_\R \dx y\, \e^{\mu_k y}\, \frac{(-1)^{k-\ell}}{2\pi \I} \int_{\I\R} \dx z\, \e^{tz^2/2-y z}\, (z-\mu_{\ell+1})\dotsb(z-\mu_k).
\end{equation}
Once again, we let run the $y$-integral over $\R_-$ and $\R_+$ separately. In the first case we take the $z$-integral such that $\Re z < \mu_k$, in the second case such that $\Re z > \mu_k$. This allows us to exchange the integrals, which gives
\begin{equation}
M_{k\ell} = \frac{(-1)^{k-\ell}}{2 \pi \I} \oint_{\Gamma_{\mu_k}} \dx z\, \e^{tz^2/2} \, \frac{(z-\mu_{\ell+1})\dotsb (z-\mu_k)}{z-\mu_k}.
\end{equation}
Since the integrand has no poles for $k>\ell$, we have $M_{k\ell}=0$ in this case, while for $k=\ell$ we get $M_{kk} = \e^{-t \mu_k^2/2}$. Thus, $M$ is upper triangular and the claim follows.
\end{proof}
With the results of Lemma~\ref{lem:1} and Lemma~\ref{lem:1b}, Theorem~\ref{thm:1} follows directly from Lemma~3.4 of \cite{BFPS06}.

We have shown that Theorem~\ref{thm:1} holds when we impose $\mu_1 < \dotsb < \mu_N$. In particular, the joint density (\ref{eq:1}) is given by an $(N(N+1)/2)$-point correlation function: With $m=N(N+1)/2$ we have
\begin{equation}\label{eq42}
(\ref{eq:1})=m!\, \rho^{(m)}(\{(\lambda_k^n,n),1\leq k \leq n \leq N\}).
\end{equation}
Let $M>0$ be any fixed real number. The density (\ref{eq:1}) is analytic in each of the $\mu_j$ in $[-M,M]$, $j=1,\dotsc,N$. The same holds for the correlation kernel (take e.g., $\mu_-=-M-1$). From this it follows that also the r.h.s.\ of (\ref{eq42}) is analytic in each of the variables $\mu_1,\dotsc,\mu_N$. Since this holds for any $M$, by analytic continuation it follows that Theorem~\ref{thm:1} holds for any given drift vector $(\mu_1,\dotsc,\mu_N)\in\R^N$.

\newpage
\section{$2+1$ dynamics with different jump rates}\label{sec:4}

In this section we show that the correlation functions \eqref{eq:24} that we obtained for the GUE matrix diffusion with drifts can be obtained as a limit from an $\GT_N$-extension of TASEP with particle-dependent jump rates. This latter process was introduced in~\cite{BF08}. Before we come to the convergence result, let us describe the model.

\begin{figure}
\begin{center}
\psfrag{x11}[cl]{\small $x_1^1$}
\psfrag{x12}[cl]{\small $x_1^2$}
\psfrag{x13}[cl]{\small $x_1^3$}
\psfrag{x14}[cl]{\small $x_1^4$}
\psfrag{x15}[cl]{\small $x_1^5$}
\psfrag{x22}[cl]{\small $x_2^2$}
\psfrag{x23}[cl]{\small $x_2^3$}
\psfrag{x24}[cl]{\small $x_2^4$}
\psfrag{x25}[cl]{\small $x_2^5$}
\psfrag{x33}[cl]{\small $x_3^3$}
\psfrag{x34}[cl]{\small $x_3^4$}
\psfrag{x35}[cl]{\small $x_3^5$}
\psfrag{x44}[cl]{\small $x_4^4$}
\psfrag{x45}[cl]{\small $x_4^5$}
\psfrag{x55}[cl]{\small $x_5^5$}
\psfrag{n=1}[cl]{\small $n=1$}
\psfrag{n=2}[cl]{\small $n=2$}
\psfrag{n=3}[cl]{\small $n=3$}
\psfrag{n=4}[cl]{\small $n=4$}
\psfrag{n=5}[cl]{\small $n=5$}
\includegraphics[width=\textwidth]{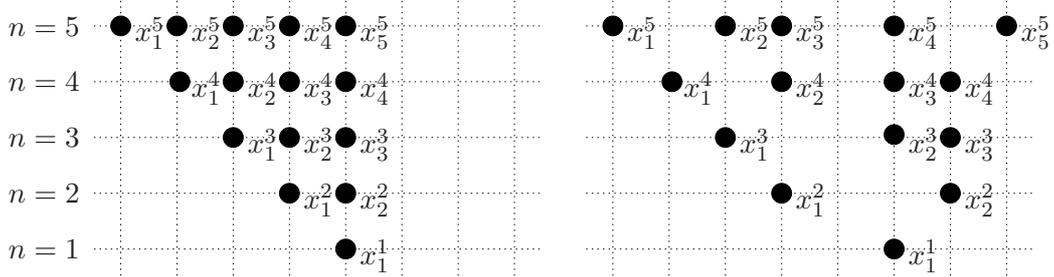}
\caption{(Left) Initial particles configuration. (Right) A possible particles configuration after some time; in this configuration, if particle $(1,3)$ tries to jump, the move is suppressed (blocked by particle $(1,2)$), while if particle $(2,2)$ jumps, then also particles $(3,3)$ and $(4,4)$ move by one unit to the right. Particles at level $n$ have a jump rate $v_n$.}
\label{FigureParticles}
\end{center}
\end{figure}

At a fixed time $t$, let us denote by $x(t)=(x_k^n(t))_{1\leq k\leq n \leq N} \in \GT_N$ the positions of the $N(N+1)/2$ particles at time $t$. We choose initial conditions $x_k^n(0)=k-n-1$ and let the particles evolve as follows: Each particle $x_k^n$ has an independent exponential clock of rate $v_n>0$, i.e., particles on the same level have the same jump rates. When the $x_k^n$-clock rings, the particle jumps to the right by one, provided that $x_k^n < x_k^{n-1}-1$, otherwise we say that $x_k^n$ is blocked by $x_k^{n-1}$. If the $x_k^n$-particle can jump, we take the largest $c\geq 1$ such that $x_k^n = x_{k+1}^{n+1}=\dotsb=x_{k+c-1}^{n+c-1}$, and all $c$ particles in this string jump to the right by one, see Figure~\ref{FigureParticles} for an example). This ensures that at any time $t$, all the particles are in $\GT_N$. More precisely, these dynamics imply that the particles stay in a discrete version of $\GT_N$, namely
\begin{equation}
\widetilde \GT_N=\{(x^1,x^2,\dotsc,x^N) \in \Z^1 \times \Z^2 \times \dotsb \times \Z^N : x_k^{n+1} < x_k^n \leq x_{k+1}^{n+1} \}.
\end{equation}

The joint distribution of the particles has been calculated in Theorem~4.1 of~\cite{BF07}, and the result is
\begin{equation}\label{eq:22}
\const \times \det\bigl[\widetilde \Psi^{N,t}_{N-k}(x_\ell^N)\bigr]_{1\leq k,\ell\leq N} \prod_{n=1}^N \det\bigl[\widetilde \phi_n(x_i^{n-1},x_j^n)\bigr]_{1\leq i,j\leq n},
\end{equation}
where
\begin{equation}
\begin{aligned}
\widetilde \Psi^{N,t}_{N-k}(x) & = \frac1{2\pi\I} \oint_{\Gamma_0}\dx z\, \e^{t/z} z^{x+N-1}(1-v_{k+1}z)\dotsb(1-v_Nz), \\
\widetilde \phi_n(x,y) & = (v_n)^{y-x}\one_{[y\geq x]} \quad \text{ and } \quad \widetilde \phi_n(x_n^{n-1},y)=(v_n)^y.
\end{aligned}
\end{equation}
Actually, Theorem~4.1 of~\cite{BF07} is a statement about the marginal of a (possibly signed) measure. However, this model is the continuous time limit of a generic Markov chain introduced in Section~2 of~\cite{BF08}, from which it follows that the measure with fully packed initial conditions $y_n = x_1^n(0)=-n$ for $1\leq n \leq N$ is actually a probability distribution. The formulation of \eqref{eq:22} follows then from the theorem by taking $a(t)=t$ and $b(t)=0$ for all $t\geq 0$. Also note that we put the transition from time $t=0$ to time $t$ (which is encoded by $\mathcal T_{t,0}$ in the theorem) into $\Psi_{N-k}^N$. As shown in~\cite{BF07}, the correlation functions of this point process are determinantal, so what remains to do is the biorthogonalization for the generic jump rates.

\begin{prop}\label{prop1}
Consider a system of particles on $\widetilde \GT_N$ with fully packed initial conditions and dynamics described above. Then, at fixed time $t$, the corresponding point process has $m$-point correlation function $\tilde \varrho^m_t$ given by
\begin{equation}
\tilde \varrho^m_t((x_1,n_1),\dotsc,(x_m,n_m)) = \det[\widetilde K^v_t((x_i,n_i),(x_j,n_j))]_{1\leq i,j\leq m}
\end{equation}
with $(x_i,n_i)\in \R \times \{1,\dotsc,N\}$ and correlation kernel
\begin{equation}
\widetilde K^v_t((x,n),(x',n'))= - \widetilde \phi^{(n,n')}(x,x')+\sum_{k=1}^{n'} \widetilde \Psi_{n-k}^{n,t}(x) \widetilde \Phi_{n'-k}^{n',t}(x'),
\end{equation}
where
\begin{align}
\widetilde \phi^{(n,n')}(x,x') & = \frac{1}{2\pi\I} \oint_{\Gamma_{0,v}} \dx z\, \frac{1}{z^{x-x'+1}}\, \frac{z^{n'-n}}{(z-v_{n+1})\dotsb(z-v_{n'})} \, \one_{[n<n']} \label{eq:27}, \\
\widetilde \Psi_{n-k}^{n,t}(x) & = \frac{1}{2\pi\I} \oint_{\Gamma_{0,v}} \dx z \, \frac{\e^{tz}}{z^{x+n+1}} \, \frac{(z-v_1)\dotsb(z-v_n)}{(z-v_1)\dotsb(z-v_k)}, \label{eq:23} \\
\widetilde \Phi_{n-\ell}^{n,t}(x) & = \frac{1}{2\pi\I} \oint_{\Gamma_v}\dx w \, \frac{w^{x+n}}{\e^{t w}} \, \frac{(w-v_1)\dotsb(w-v_{\ell-1})}{(w-v_1)\dotsb(w-v_n)}.
\end{align}
\end{prop}

\begin{proof}
By Proposition~3.1 of~\cite{BF07}, we have
\begin{equation}
\widetilde \Psi_{n-k}^{n,t}(x) = \frac1{2\pi\I} \oint_{\Gamma_0} \dx w\, w^{x+k-1} \e^{t/w} \, \frac{(1-v_1w)\dotsb(1-v_nw)}{(1-v_1w)\dotsb(1-v_kw)}
\end{equation}
for $k\geq 1$. A change of variable $z=1/w$ then yields \eqref{eq:23}. Next we need to verify that $\{\widetilde \Psi_{n-k}^{n,t}:1\leq k \leq n\}$ is biorthogonal to $\{\widetilde \Phi_{n-\ell}^{n,t}:1\leq \ell\leq n\}$ (see Eq.~(3.5) of~\cite{BF07}). We split the sum over $\Z$ into two parts, one over $x\geq 0$ and one over $x<0$. Then,
\begin{multline}
\sum_{x \geq 0} \widetilde \Psi_{n-k}^{n,t}(x)\widetilde\Phi_{n-k}^{n,t}(x) \\
= \sum_{x \geq 0} \frac{1}{(2\pi\I)^2} \oint_{\Gamma_v}\dx w\oint_{\Gamma_0}\dx z\, \frac{\e^{t z}}{\e^{t w}}\, \frac{w^{x+n}}{z^{x+n+1}} \, \frac{(z-v_{k+1})\dotsb(z-v_n)}{(w-v_\ell)\dotsb(w-v_n)}.
\end{multline}
We choose $\Gamma_0$ and $\Gamma_v$ such that $\lvert w \rvert \leq \lvert z \rvert$ which allows us to put the sum inside the integrals. This gives
\begin{multline}
\sum_{x \geq 0} \widetilde\Psi_{n-k}^{n,t}(x)\widetilde\Phi_{n-k}^{n,t}(x)\\
 = \frac{1}{(2\pi\I)^2} \oint_{\Gamma_v}\dx w \oint_{\Gamma_{0,w}} \dx z\, \frac{\e^{tz}}{\e^{tw}}\,\frac{w^n}{z^n} \, \frac{(z-v_{k+1})\dotsb(z-v_n)}{(w-v_k)\dotsb(w-v_\ell)} \, \frac{1}{z-w}.
\end{multline}
For $x<0$ we choose $\Gamma_0$ and $\Gamma_v$ such that they satisfy $\lvert w \rvert > \lvert z \rvert$ which gives
\begin{multline}
\sum_{x < 0} \widetilde\Psi_{n-k}^{n,t}(x)\widetilde\Phi_{n-k}^{n,t}(x) \\
 = - \frac{1}{(2\pi\I)^2} \oint_{\Gamma_0}\dx z \oint_{\Gamma_{v,z}} \dx w\, \frac{\e^{tz}}{\e^{tw}}\,\frac{w^n}{z^n} \, \frac{(z-v_{k+1})\dotsb(z-v_n)}{(w-v_k)\dotsb(w-v_\ell)} \, \frac{1}{z-w}.
\end{multline}
Thus,
\begin{equation}
\sum_{x\in\Z} \widetilde\Psi_{n-k}^{n,t}(x)\widetilde\Phi_{n-k}^{n,t}(x) = \frac{1}{2\pi\I} \oint_{\Gamma_v} \dx w \, \frac{(w-v_{k+1})\dotsb(w-v_n)}{(w-v_\ell)\dotsb(w-v_n)} = \delta_{k\ell}.
\end{equation}
Finally, we show that $\{\widetilde \Phi_{n-\ell}^{n,t}:1\leq \ell\leq n\}$ spans the space of functions $V_n$. We denote by $u_1<\dotsb<u_\nu$ the different values of $v_1,\dotsc,v_n$ and $\alpha_k$ the multiplicity of $u_k$, i.e., $\alpha_1+\dotsb+\alpha_\nu = n$. Then, we may write
\begin{equation}
\begin{aligned}
\widetilde \Phi_{n-1}^{n,t}(x) & = \frac{1}{2\pi\I} \oint_{\Gamma_v}\dx w \, \frac{w^{x+n}}{\e^{t w}} \, \frac{1}{(w-u_1)^{\alpha_1}\dotsb(w-u_\nu)^{\alpha_\nu}} \\
& = \sum_{i=1}^\nu\frac1{(\alpha_i-1)!}\,\frac{\dx^{\alpha_i-1}}{\dx w^{\alpha_i-1}}\bigg|_{w=u_i} \biggl( \frac{w^{x+n}}{\e^{t w}} \prod_{j \neq i} \frac1{(w-u_j)^{\alpha_j}}\biggr) \\
& = \sum_{i=1}^\nu (u_i)^x \sum_{j=1}^{\alpha_i} c_{i,j} x^{j-1}.
\end{aligned}
\end{equation}
For $\ell=2,\dotsc,n$, we can represent $\widetilde \Phi_{n-\ell}^{n,t}$ in the same way, but with exponents $\alpha_{\ell,i} \le\alpha_i$, $1\leq i\leq \nu$. Since $(\alpha_{k,1},\dotsc,\alpha_{k,\nu}) \neq (\alpha_{\ell,1},\dotsc,\alpha_{\ell,\nu})$ for $k\neq \ell$, this shows that
\begin{equation}
\lin \{\widetilde \Phi_{n-\ell}^{n,t} : 1\leq \ell \leq n\} = \lin \{x\mapsto (u_i)^x x^{j-1} :1\leq u \leq \nu, 1 \leq j \leq \alpha_i\},
\end{equation}
which is $V_n$.
\end{proof}

We continue by establishing the convergence result under the scaling (\ref{scaling}). Correspondingly, we rescale (and conjugate) the kernel $\widetilde K_t$ and define the rescaled kernel as
\begin{equation}
K_{\tau,T,\mathrm{resc}}^\mu((\xi,n),(\xi',n'))
=\frac{T^{n'/2}}{T^{n/2}} \, \sqrt T \, \widetilde K_{\tau T}^{\mu_T}\bigl(([\tau T - \xi \sqrt{T}],n),([\tau T-\xi' \sqrt{T}],n')\bigr)
\end{equation}
where $[\,\cdot\,]$ denotes the integer part, and the drift $v$ is now $\mu_T=1-\mu/\sqrt T$. Of course, $T$ is assumed to be so large that $\mu_T>0$ is satisfied.

\begin{prop}\label{prop:1}
For any fixed $L>0$, the rescaled kernel $K_{\tau,T,\mathrm{resc}}^\mu$ converges, uniformly for $\xi,\xi'\in[-L,L]$, as
\begin{equation}
\lim_{T\to\infty} K_{\tau,T,\mathrm{resc}}^\mu((\xi,n),(\xi',n')) = K_\tau^\mu((\xi,n),(\xi',n'))
\end{equation}
with $K_\tau^\mu\equiv K_\tau$ given in \eqref{eq:26}.
\end{prop}

\begin{proof}
Let us define the rescaled functions
\begin{equation}
\begin{aligned}
\phi^{(n,n')}_{T,\mathrm{resc}}(\xi,\xi') & = T^{-(n'-n+1)/2} \, \widetilde \phi^{(n,n')}(\tau T + \xi \sqrt T,\tau T + \xi' \sqrt T), \\
\Psi_{n-k, T,\mathrm{resc}}^{n,\tau}(\xi) & = T^{(n-k+1)/2} \e^{-\tau T} \, \widetilde \Psi_{n-k}^{n,\tau T}(\tau T + \xi \sqrt T), \\
\Phi_{n-k, T,\mathrm{resc}}^{n,\tau}(\xi') & = T^{-(n-k)/2} \e^{\tau T} \, \widetilde \Phi_{n-k}^{n,\tau T}(\tau T + \xi' \sqrt T),
\end{aligned}
\end{equation}
where we also rescale the jump rates as in \eqref{scaling}. We have to show that these functions converge to their analogues from \eqref{eq:9}--\eqref{eq:9b}. We first verify that $\phi^{(n,n')}_{T,\mathrm{resc}}(\xi,\xi')\to \phi^{(n,n')}(\xi,\xi')$ with $n<n'$. For $y\geq y'$, the integrand of $\widetilde \phi^{(n,n')}(y,y')$ in \eqref{eq:27} has residue $0$ at infinity and thus the whole integral vanishes, while for $y<y'$, there is no pole at $z=0$ and therefore
\begin{equation}
\widetilde \phi^{(n,n')}(y,y')=\sum_{i=n+1}^{n'} v_i^{(y'-y)+(n'-n)-1} \prod_{j\neq i} \frac{1}{v_i-v_j} \, \one_{[y<y']}.
\end{equation}
Hence, for its rescaled version,
\begin{equation}
\phi_{T,\mathrm{resc}}^{(n,n')}(\xi,\xi')= \sum_{i=n+1}^{n'} \left( 1-\frac{\mu_i}{\sqrt T }\right)^{(\xi-\xi')\sqrt T +(n'-n)-1} \prod_{j\neq i}\frac{1}{\mu_j-\mu_i}\,\one_{[\xi>\xi']},
\end{equation}
which, as $T\to \infty$, converges to
\begin{equation}
\sum_{i=n+1}^{n'} \e^{\mu_i(\xi'-\xi)} \prod_{j\neq i} \frac{1}{\mu_j-\mu_i}\one_{\{\xi>\xi'\}} = \frac{(-1)^{n'-n}}{2\pi\I} \int_{\I\R+\mu_-} \hspace{-1em}\dx z \, \frac{\e^{z(\xi'-\xi)}}{(z-\mu_{n+1})\dotsb(z-\mu_{n'})}.
\end{equation}

Next we show that $\Psi_{n-k, T,\mathrm{resc}}^{n,\tau}(\xi)\to \Psi_{n-k}^{n,\tau}(\xi)$ uniformly for $\xi \in [-L,L]$. We have
\begin{equation}
\begin{aligned}\label{eq:28}
\Psi_{n-k, \mathrm{resc}}^{n,\tau,T}(\xi) & = \frac{\sqrt T}{2\pi \I} \oint_{\Gamma_{0,v}} \dx z \, \frac{\e^{\tau T(z-1)}}{z^{\tau T-\xi\sqrt T+n+1}} \, g(z) \\
& = \frac{\sqrt T}{2\pi \I} \oint_{\Gamma_{0,v}} \dx z\, \e^{\tau T f_0(z)+\sqrt T f_1(z) +f_2(z)} g(z)
\end{aligned}
\end{equation}
with $f_0(z)=z-1-\ln z$, $f_1(z)=\xi \ln z$, $f_2(z)=-(n+1)\ln z$, and
\begin{equation}
g(z)= \frac{(\sqrt T(z-1)+\mu_1)\dotsb (\sqrt T(z-1)+\mu_n)}{(\sqrt T(z-1)+\mu_1)\dotsb(\sqrt T(z-1)+\mu_k)}.
\end{equation}
A Taylor expansion around the double critical point of $f_0$, i.e., around $z_c=1$ gives
\begin{equation}
\begin{aligned}
f_0(z) & = \frac12(z-1)^2 + \Or((z-1)^3), \\
f_1(z) & = \xi(z-1) + \Or((z-1)^2), \\
f_2(z) & = 0+\Or(z-1).
\end{aligned}
\end{equation}

Fix $r>1-\mu_-/\sqrt{T}$ and deform $\Gamma_{0,v}$ to the contour $\gamma= \gamma_1 \cup \gamma_2$ with
\begin{equation}
\gamma_1 = 1-\mu_-/\sqrt{T} + \I [-r,r], \quad \gamma_2 = \{\lvert z \rvert = r\} \cap \{\Re z < 1-\mu_-/\sqrt{T}\}.
\end{equation}
Let us verify that $\gamma$ is a steep descent path for $f_0$. On the segment $\gamma_1$, we have that $\Re f_0(x+\I y)=x-1-\frac12\ln(x^2+y^2)$, so
\begin{equation}
\frac{\dx \Re f_0(x+\I y)}{\dx y}  = - \frac{y}{x^2+y^2}, \qquad y \in [-r,r],
\end{equation}
with $x=1-\mu_-/\sqrt{T}$. Thus $f_0$ is strictly increasing on $1-\frac{\mu_-}{\sqrt T}+\I [-r,0)$ and strictly decreasing on $1-\frac{\mu_-}{\sqrt T}+\I (0,r]$. On $\gamma_2$ we compute
\begin{equation}
\Re f_0(r \e^{\I \varphi}) = r \cos \varphi -1 -\ln r, \qquad \varphi \in (\arccos \tfrac1r,2\pi-\arccos \tfrac1r),
\end{equation}
which means that $f_0$ is strictly decreasing on $\gamma_2 \cap \{\Im z>0\}$ and strictly increasing on $\gamma_2 \cap \{\Im z < 0\}$. Thus $\gamma$ is a steepest descent path for $f_0$ and the major contribution comes from a line segment $\gamma_\delta = 1-\frac{\mu_-}{\sqrt T}+\I [-\delta,\delta]$ for any $\delta \in (0,1)$. Indeed, the error we make when we integrate along $\gamma_\delta$ instead of $\gamma$ is of order $\Or(\e^{-c T})$ with $c\sim \delta^2$. We therefore consider the integral on $\gamma_\delta$ only,
\begin{equation}\label{eq:29}
\frac{\sqrt T}{2\pi\I} \int_{\gamma_\delta} \dx z \,  g(z) \e^{\xi \sqrt T (z-1)+\frac {\tau T} 2 (z-1)^2} \e^{\Or \left(z-1,\sqrt T (z-1)^2,T (z-1)^3\right)}.
\end{equation}
Using $\lvert \e^x-1\rvert \leq \lvert x \rvert \e^{\lvert x \rvert}$, the difference between \eqref{eq:29} and the same integral without the error term can be bounded by
\begin{multline}
\frac{\sqrt T}{2\pi}\int_{\gamma_\delta} \dx z\, \Big| \e^{c_1 \xi \sqrt T (z-1)+c_2\frac {\tau T} 2 (z-1)^2}\\
\times \Or \left(z-1,\sqrt T (z-1)^2,T(z-1)^3,T^{(n-k)/2}(z-1)^{n-k}\right)\Big|
\end{multline}
for some constants $c_1$ and $c_2$ that can be chosen arbitrarily close to $1$ as $\delta \to 0$. By a change of variable $Z = \sqrt T(1-z)$ one then sees that this error is of order $\Or(T^{-1/2})$. Hence we can consider the integral in \eqref{eq:29} without the error term, which simplifies to
\begin{equation}
\frac{\sqrt T}{2\pi\I} \int_{\gamma_\delta} \dx z\,\e^{\tau T(z-1)^2/2+\xi \sqrt T(z-1)}g(z).
\end{equation}
The error we make if we extend $\gamma_\delta$ to $1-\frac{\mu_-}{\sqrt T}+\I \R$ is of order $\Or(\e^{-c T})$. All together the integral from \eqref{eq:28} agrees, up to an error $\Or(\e^{-c T},T^{-1/2})$ uniform in $\xi \in [-L,L]$, with
\begin{equation}
\frac{\sqrt T}{2\pi\I} \int_{1-\frac{\mu_-}{\sqrt T}+\I \R} \dx z \, \e^{\tau T (z-1)^2/2+\xi \sqrt T(z-1)} g(z),
\end{equation}
where the poles of $g$ lie on the left of the integration axis. After a change of variable $Z=-\sqrt T(z-1)$ this integral can be identified as $\Psi_{n-k}^{n,\tau}(\xi)$.

Finally we show that $\Phi_{n-k,T,\mathrm{resc}}^{n,\tau}(\xi') \to \Phi_{n-k}^{n,\tau}(\xi')$. We have
\begin{multline}
\Phi_{n-k, T,\mathrm{resc}}^{n,\tau}(\xi') = \frac{\sqrt T}{2\pi\I} \oint_{\Gamma_v}\dx w \, \e^{\tau T(\ln w-w+1)-\sqrt T \xi' \ln w + n \ln w} \\
\times \frac{(\sqrt T(w-1)+\mu_1)\dotsb(\sqrt T(w-1)+\mu_{k-1})}{(\sqrt T(w-1)+\mu_1)\dotsb(\sqrt T(w-1)+\mu_n)},
\end{multline}
and by a change of variable $W=-\sqrt T(w-1)$ and a Taylor expansion in the exponent we get
\begin{multline}
\Phi_{n-k,T,\mathrm{resc}}^{n,\tau}(\xi') \\
= \frac{(-1)^{n-k+1}}{2\pi\I} \oint_{\Gamma_a} \dx W\, \e^{-\tau W^2/2+\xi' W+\Or(T^{-1/2})} \, \frac{(w-\mu_1)\dotsb (w-\mu_{k-1})}{(w-\mu_1)\dotsb(w-\mu_n)},
\end{multline}
which converges uniformly for $\xi'\in[-L,L]$ to $\Phi_{n-k}^{n,\tau}(\xi')$.
\end{proof}

With the above results we can now prove Theorem~\ref{thm:3}.

\begin{proof}[Proof of Theorem~\ref{thm:3}]
Set $m=N(N+1)/2$ and define $n_1,\dotsc,n_m$ by
\begin{equation}
n_1=1, \enspace n_2=n_3=2, \enspace n_4=n_5=n_6 = 3, \enspace \dotsc, \enspace n_{m-N+1}=\dotsb = n_m =N.
\end{equation}
For $A \subseteq \R^m$ we set $A_T = (\tau T-\sqrt T A)\cap \Z$. Then, we have
\begin{equation}
\begin{aligned}
\nu_T&(A) = \sum_{(x_1,\dotsc,x_m) \in A_T} \det\bigl[ \widetilde K_{\tau T}^{\mu_T} ((x_i,n_i),(x_j,n_j)) \bigr]_{1\leq i,j\leq m} \\
& = T^{m/2} \int_A\dx^m x \det\bigl[\widetilde K_{\tau T}^{\mu_T} ((\tau T - [x_i] \sqrt T,n_i),(\tau T-[x_j]\sqrt T,n_j))\bigr]_{1\leq i,j\leq m} \\
& = \int_A\dx^m x \det\bigl[\widetilde K_{\tau,T,\mathrm{resc}}(([x_i],n_i),([x_j],n_j))\bigr]_{1\leq i,j\leq m}
\end{aligned}
\end{equation}
and
\begin{equation}
\nu(A) = \int_A \dx^m x \det\bigl[K_\tau^\mu((x_i,n_i),(x_j,n_j))\bigr]_{1\leq i,j\leq m}.
\end{equation}
Since the determinants are continuous functions of the kernels, we have by Proposition~\ref{prop:1} that
\begin{multline}
\lim_{T\to \infty}\det\bigl[\widetilde K^\mu_{\tau,T,\mathrm{resc}}(([x_i],n_i),([x_j],n_j))\bigr]_{1\leq i,j\leq m}\\
= \det\bigl[K_\tau^\mu((x_i,n_i),(x_j,n_j))\bigr]_{1\leq i,j\leq m}
\end{multline}
for all $x_1,\dotsc,x_m \in \R$. Thus we have shown that the densities of the probability measures in question converge pointwise to each other. Then, \eqref{eq:31} is a direct consequence of Scheff\'e's theorem, see e.g.~\cite{Bil68}.
\end{proof}

\newpage
\section{Warren's process with drifts}\label{sec:3}
We have seen in Section~\ref{sec:2} that the eigenvalues' density can be written as a product of determinants, and, in Lemma~\ref{lem:2}, we calculated the normalization constant, so that the probability measure on the eigenvalues reads
\begin{equation}\label{eq:19}
\Pb\biggl( \bigcap_{1\leq k\leq n \leq N} \{\lambda_k^n \in \dx \lambda_k^n\}\biggr) = \tilde p_t(\lambda) \, \dx\lambda
\end{equation}
with $\dx\lambda=\prod_{1\leq k\leq n\leq N} \dx \lambda^n_k$, and
\begin{equation}\label{eqDensity}
\tilde p_t(\lambda)= \det\bigl[\Psi_{N-k}^{N,t}(\lambda_\ell^N)\bigr]_{1\leq k,\ell \leq N} \prod_{n=1}^N \e^{-t \mu_n^2/2} \prod_{n=1}^N \det\bigl[\phi_n(\lambda_i^{n-1},\lambda_j^n)\bigr]_{1\leq i,j\leq n}
\end{equation}
In this section we explain the connection to a system of Brownian motions in $\GT_N$. More precisely, we consider Brownian motions $\{B_k^n, 1\leq k \leq n \leq N\}$ in $\GT_N$ starting from $0$, with drift $\mu_n$, and interacting as follows:
\begin{itemize}
\item The evolution of $B_k^n$ does not depend on the Brownian motions with higher upper index ($B_\ell^m$ for $m\geq n+1$, and any $\ell$);
\item $B_k^{n}$ is reflected off $B_k^{n-1}$ and $B_{k-1}^{n-1}$.
\end{itemize}
These reflections are sometimes called oblique reflections~\cite{WW84}, since in the $(x_k^{n-1},x_k^n)$-plane (resp.\ $(x_{k-1}^{n-1},x_k^n)$-plane) the reflection directions are not normal, but oblique as indicated in Figure~\ref{FigReflections}. Note that the projection on $\{B_1^n,1\leq n \leq N\}$ differs from the process studied in~\cite{PP08}, where the reflections are in the normal direction.
\begin{figure}
\begin{minipage}{.45\textwidth}
\begin{center}
\includegraphics{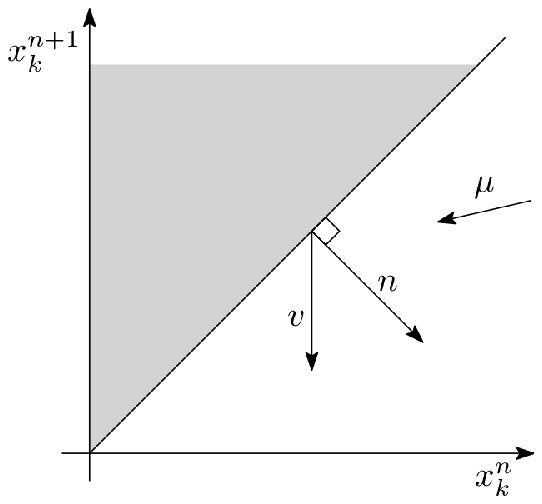}
\end{center}
\end{minipage}
\hfill
\begin{minipage}{.45\textwidth}
\begin{center}
\includegraphics{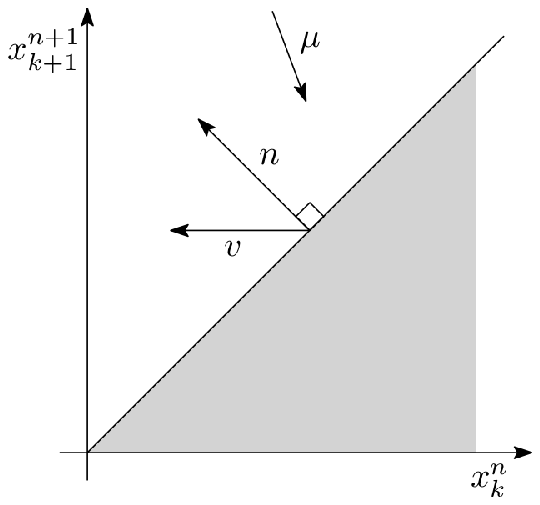}
\end{center}
\end{minipage}
\caption{The two reflection types in our system. They correspond to the boundary condition (\ref{eq:21}).}
\label{FigReflections}
\end{figure}

Let us now describe the system of Brownian motions. Denote by $p_t$ be the probability density of the Brownian motions in $\GT_N$ (its existence will be a consequence of our result). Following~\cite{HW87}, where Brownian motions with oblique reflections were studied, for a Brownian motion with drift $\mu$ reflected at the boundary in the direction $v$, the boundary conditions on the density function may be expressed as follows. Denote by $n$ the normal vector of the boundary, let $v$ be normalized such that $n\cdot v=1$ and let $q=v-n$. Moreover, set $\nabla_T=\nabla - n (n\cdot \nabla)$, $D^*=n\cdot \nabla - q\cdot \nabla_T$. Then, the boundary condition can be written as
\begin{equation}
D^* p_t = (\nabla_T \cdot q + 2 \mu \cdot n) p_t \quad \textrm{on the boundary}.
\end{equation}
Specializing to our case, we get
\begin{equation}\label{eq:21}
\frac{\partial}{\partial x_k^n} \, p_t(x) + (\mu_{n+1}-\mu_n)p_t(x)=0,
\end{equation}
whenever $x_k^n=x_k^{n+1}$ or $x_k^n=x_{k+1}^{n+1}$, for $1\leq k\leq N-1$.

This process, without drifts, was introduced by Warren in~\cite{War07}, where he determined the transition probability for any initial condition and also showed that the process is well-defined when starting from $0$. We here consider a system of Brownian motions with constant (bounded) drifts, which can be expressed as follows,
\begin{equation}
\begin{aligned}
B_1^1(t)&=\mu_1 t + b_1^1(t),\\
B_1^n(t)&=\mu_n t + b_1^n(t)-L_{B_1^{n-1}-B_1^n}(t),\quad n=2,\dotsc,N,\\
B_k^n(t)&=\mu_n t + b_k^n(t)-L_{B_k^{n-1}-B_k^n}(t)+L_{B_k^n-B_{k-1}^{n-1}}(t),\quad 2\leq k < n \leq N,\\
B_n^n(t)&=\mu_n t + b_n^n(t)+L_{B_n^n-B_{n-1}^{n-1}}(t),\quad n=2,\dotsc,N,
\end{aligned}
\end{equation}
where the $b_k^n$, $1\leq k \leq n \leq N$, are independent standard Brownian motions and $L_{X-Y}(t)$ is twice the semimartingale local time at zero of $X(t)-Y(t)$. The question of well-definedness was related to the, a priori possible, presence of triple collisions. Bounded drifts do not influence this property as can be seen by applying Girsanov's theorem like in the works~\cite{IK10,KPS12}.

Reflected Brownian motions can be also defined as follows. A standard one-dimensional reflected Brownian motion can also be defined to be the image under the Skorokhod map of standard Brownian motion. More precisely, one define a Brownian motion, $B(t)$, starting from $y\in\R$ and being reflected at some continuous function $f(t)$ with $f(0)<y$ is via the Skorokhod representation~\cite{Sko61,AO76}
\begin{equation}\begin{aligned}
B(t)&=y+b(t)-\min\big\{0,\inf_{0\leq s \leq t}(y+b(s)-f(s))\big\} \\
&= \max\big\{y+b(t),\sup_{0\leq s \leq t}(f(s)+b(t)-b(s))\big\},
\end{aligned}\end{equation}
where $b$ is a standard Brownian motion starting at $0$. In this paper we use as definition of the Warren process with drifts to be the image of independent Brownian motions under the extended Skorokhod map introduced by Burdzy, Kang and Ramanan (see Theorem~2.6 of~\cite{BKR09} for an explicit formula).

\begin{proof}[Proof of Theorem~\ref{thm:2}]
Consider a particle system as in Section~\ref{sec:4} but where the particles evolves independently, i.e., $\tilde x_k^n(0)=-n+k-1$ for $1\leq k \leq n\leq N$ and the evolution of $\tilde x_k^n(t)$ is a continuous time random walk with jump rate $v_n$. Consider now the scaling (\ref{scaling})
\begin{equation}
t=\tau T,\quad \tilde B_k^n=\frac{\tilde x_k^n-\tau T}{-\sqrt{T}}, \quad v_n=1-\frac{\mu_n}{\sqrt{T}}.
\end{equation}
The $\tilde x_k^n$'s are independent, so in the $T\to\infty$ limit, \mbox{$(\tilde B_k^n, 1\leq k\leq n\leq N)$} converges weakly to a \mbox{$N(N+1)/2$}-dimensional Brownian motion \mbox{$(B_k^n, 1\leq k\leq n\leq N)$}, where $B_k^n$ has drift $\mu_n$ (see Donsker's theorem). As shown in~\cite{GS12} by Gorin and Shkolnikov, the particle with the blocking/pushing dynamics converges wearkly as $T\to\infty$ to the Warren process with level-dependent drifts. To be precise, they first showed the convergence for the drift-less case, where the limit process is the Warren process. However the same proof applies for more generic cases including the one of this paper, see Remark 10 of~\cite{GS12}.

In Proposition~\ref{prop1} we have proven that the correlation functions has a limit as $T\to\infty$. Further, the integral of the density is one, so that no mass is lost at infinity or localized in some Dirac mass. Thus, the $n$-point correlation function of the reflected Brownian motion is the $T\to\infty$ limit of the $n$-point correlation function for the interacting particle system.
\end{proof}

For completeness, let us remark that the transition density $p_t$ in $\GT_N$ satisfies:\\[0.5em]
(1) the Fokker-Planck equation (or Kolmogorov forward equation)
\begin{equation}\label{eq:16}
\frac{\partial}{\partial t}\, p_t(x)=\sum_{n=1}^N \sum_{k=1}^n \left( \frac{1}{2}\, \frac{\partial^2}{\partial (x_k^n)^2}-\mu_n \frac{\partial}{\partial x_k^n}\right)  p_t(x),
\end{equation}
(2) the initial condition
\begin{equation}\label{eq:20}
\lim_{t\searrow 0} p_t(x) \dx x = \prod_{1\leq k\leq n\leq N} \delta_{x_k^n},
\end{equation}
(3) the boundary condition (\ref{eq:21}).

\begin{prop}\label{prop:2}
Denote by $p_t: \GT_N \to [0,1]$ be the probability density defined in (\ref{eqDensity}). Inside $\GT_N$, this density satisfies the Fokker-Planck equation (\ref{eq:16}), the initial condition (\ref{eq:20}), and the boundary condition (\ref{eq:21}).
\end{prop}
\begin{proof}
First observe that by setting $\tilde \Psi_{N-k}^{N,t}(x) = \e^{\mu_N x} \Psi_{N-k}^{N,t}(x)$, we can rewrite \eqref{eqDensity} as a probability measure on $\GT_N$ with density
\begin{equation}
\tilde p_t(x) = \det\left[ \tilde \Psi_{N-k}^{N,t}(x_\ell^N)\right]_{1\leq k,\ell\leq N} \prod_{k=1}^N \e^{-t\mu_k^2/2} \prod_{n=1}^{N-1} \prod_{k=1}^n \e^{(\mu_n-\mu_{n+1})x_k^n}.
\end{equation}
for $x=(x_k^n)_{1\leq k\leq n \leq N}\in \GT_N$. The double product only depends on $(x_k^n)_{1\leq k\leq n\leq N-1}$, while the determinant is a function of $(x_k^N)_{1\leq k\leq N}$.
We have
\begin{equation}
\frac{1}{2} \frac{\partial^2}{\partial x^2}\,\tilde \Psi_{N-k}^{N,t}(x)= \frac{\partial}{\partial t}\,\tilde \Psi_{N-k}^{N,t}(x)+\mu_N \frac{\partial}{\partial x} \tilde \Psi_{N-k}^{N,t}(x) -\frac{\mu_N^2}{2} \tilde \Psi_{N-k}^{N,t}(x),
\end{equation}
from which follows that
\begin{equation} \label{eq:13}
\frac{1}{2} \sum_{\ell=1}^N \frac{\partial^2}{\partial (x_\ell^N)^2}\,\tilde p_t(x) = \frac{\partial}{\partial t}\,\tilde p_t(x)+\mu_N \sum_{\ell=1}^N \frac{\partial}{\partial x_\ell^N}\,\tilde p_t(x)+\frac12 \left(\sum_{n=1}^N \mu_n^2-N \mu_N^2\right)\tilde p_t(x).
\end{equation}
For $k=1,\dots,N-1$, we have
\begin{equation}\label{eq:14}
\frac{\partial}{\partial x_k^n} \, \tilde p_t(x)=(\mu_n-\mu_{n+1})\tilde p_t(x), \quad \frac{\partial^2}{\partial (x_k^n)^2} \, \tilde p_t(x)=(\mu_n-\mu_{n+1})^2 \tilde p_t(x),
\end{equation}
and thus, putting \eqref{eq:13} and \eqref{eq:14} together,
\begin{multline} \label{eq:15}
\frac{1}{2} \sum_{n=1}^N\sum_{k=1}^n \frac{\partial^2}{\partial(x_k^n)^2}\,\tilde p_t(x) = \frac{\partial}{\partial t}\,\tilde p_t(x) +\mu_N \sum_{k=1}^N \frac{\partial}{\partial x_k^N}\,\tilde p_t(x) \\
+ \frac{1}{2} \left(\sum_{n=1}^N \mu_n^2-N \mu_N^2+ \sum_{n=1}^{N-1} n(\mu_n-\mu_{n+1})^2 \right)\tilde p_t(x).
\end{multline}
Using that
\begin{equation} \label{eq:17}
N \mu_N^2-\sum_{n=1}^N \mu_n^2 = \sum_{n=1}^{N-1}n(\mu_{n+1}^2-\mu_n^2) = \sum_{n=1}^{N-1}n(\mu_n-\mu_{n+1})^2 - 2 \sum_{k=1}^{N-1} n \mu_n(\mu_n-\mu_{n+1})
\end{equation}
the expression between the brackets in \eqref{eq:15} simplifies to $2\sum n\mu_n(\mu_n-\mu_{n+1})$.
On the other hand,
\begin{equation} \label{eq:18}
\sum_{n=1}^N \sum_{k=1}^n \mu_n \frac{\partial}{\partial x_k^n} \,\tilde p_t(x)= \sum_{n=1}^{N-1}n\mu_n(\mu_n-\mu_{n+1})\tilde p_t(x)+\mu_N \sum_{k=1}^N \frac{\partial}{\partial x_k^N} \, \tilde p_t(x).
\end{equation}
Then, \eqref{eq:16} follows from \eqref{eq:15}, \eqref{eq:17} and \eqref{eq:18}. The initial condition \eqref{eq:20} is satified because as $t\searrow 0$, we obtain the Dirac measure at $x_k^N=0$ for $1\leq k \leq N$ and since we consider $\tilde p_t$ on $\GT_N$, this immediately implies that $x^k_n=0$ for all $1\leq k \leq n \leq N-1$. Finally, the boundary condition \eqref{eq:21} holds trivially by \eqref{eq:14}.
\end{proof}

\begin{remark}
The three conditions in Proposition~\ref{prop:2} are, in general, not enough to prove that $\tilde p_t=p_t$. For that, one would need the backwards equation.
\end{remark}

\appendix

\section{Determinantal correlations}
Since we refer several times to Lemma~3.4 of~\cite{BFPS06}, we report it here.
\begin{lem}[Lemma~3.4 of~\cite{BFPS06}]
Assume we have a signed measure on $\{x_i^n,n=1,\dotsc,N,i=1,\dotsc,n\}$ given in the form,
\begin{equation}\label{Sasweight}
 \frac{1}{Z_N}\prod_{n=1}^{N-1} \det[\phi_n(x_i^n,x_j^{n+1})]_{1\leq i,j\leq n+1} \det[\Psi_{N-i}^{N}(x_{j}^N)]_{1\leq i,j \leq N},
\end{equation}
where $x_{n+1}^n$ are some ``virtual'' variables and $Z_N$ is a normalization constant. If $Z_N\neq 0$, then the correlation functions are determinantal.

To write down the kernel we need to introduce some notations. Define
\begin{equation}\label{Sasdef phi12}
\phi^{(n_1,n_2)}(x,y)=
\begin{cases} (\phi_{n_1} \ast \dotsb \ast \phi_{n_2-1})(x,y),& n_1<n_2,\\
0,& n_1\geq n_2,
\end{cases}.
\end{equation}
where $(a* b)(x,y)=\sum_{z\in\Z}a(x,z) b(z,y)$, and, for $1\leq n<N$,
\begin{equation}\label{Sasdef_psi}
\Psi_{n-j}^{n}(x) := (\phi^{(n,N)} * \Psi_{N-j}^{N})(y), \quad j=1,\dotsc,N.
\end{equation}
Set $\phi_0(x_1^0,x)=1$. Then the functions
\begin{equation}
\{ (\phi_0*\phi^{(1,n)})(x_1^0,x), \dots,(\phi_{n-2}*\phi^{(n-1,n)})(x_{n-1}^{n-2},x), \phi_{n-1}(x_{n}^{n-1},x)\}
\end{equation}
are linearly independent and generate the $n$-dimensional space $V_n$. Define a set of functions $\{\Phi_j^{n}(x), j=0,\dotsc,n-1\}$ spanning $V_n$ defined by the orthogonality relations
\begin{equation}\label{Sasortho}
\sum_x \Phi_i^n(x) \Psi_j^n(x) = \delta_{i,j}
\end{equation}
for $0\leq i,j\leq n-1$.

Further, if $\phi_n(x_{n+1}^n,x)=c_n \Phi_0^{(n+1)}(x)$, for some $c_n\neq 0$, \mbox{$n=1,\dotsc,N-1$}, then the kernel takes the simple form
\begin{equation}\label{SasK}
K(n_1,x_1;n_2,x_2)= -\phi^{(n_1,n_2)}(x_1,x_2)+ \sum_{k=1}^{n_2} \Psi_{n_1-k}^{n_1}(x_1) \Phi_{n_2-k}^{n_2}(x_2).
\end{equation}
\end{lem}


\begin{thebibliography}{10}

\bibitem{ANvM10b}
M.~Adler, E.~Nordenstam, and P.~van Moerbeke, \emph{{Consecutive Minors for
  Dyson's Brownian Motions}}, arXiv:1007.0220 (2010).

\bibitem{AvMW13}
M.~Adler, P.~van Moerbeke, and D.~Wang, \emph{Random matrix minor processes
  related to percolation theory}, arXiv:1301.7017 (2013).

\bibitem{AO76}
R.F. Anderson and S.~Orey, \emph{Small random perturbation of dynamical systems
  with reflecting boundary}, Nagoya Math. J. \textbf{60} (1976), 189--216.

\bibitem{Bar01}
Y.~Baryshnikov, \emph{{GUEs and queues}}, Probab. Theory Relat. Fields
  \textbf{119} (2001), 256--274.

\bibitem{Bil68}
P.~Billingsley, \emph{{Convergence of Probability Measures}}, Wiley ed., New
  York, 1968.

\bibitem{BF08}
A.~Borodin and P.L. Ferrari, \emph{{Anisotropic Growth of Random Surfaces in
  $2+1$ Dimensions}}, To appear in Comm. Math. Phys.; arXiv:0804.3035 (2008).

\bibitem{BF07}
A.~Borodin and P.L. Ferrari, \emph{{Large time asymptotics of growth models on
  space-like paths I: PushASEP}}, Electron. J. Probab. \textbf{13} (2008),
  1380--1418.

\bibitem{BFPS06}
A.~Borodin, P.L. Ferrari, M.~Pr{\"a}hofer, and T.~Sasamoto, \emph{{Fluctuation
  properties of the TASEP with periodic initial configuration}}, J. Stat. Phys.
  \textbf{129} (2007), 1055--1080.

\bibitem{BG08}
A.~Borodin and V.~Gorin, \emph{Shuffling algorithm for boxed plane partitions},
  Adv. Math. \textbf{220} (2009), 1739--1770.

\bibitem{BK07}
A.~Borodin and J.~Kuan, \emph{{Asymptotics of Plancherel measures for the
  infinite-dimensional unitary group}}, Adv. Math. \textbf{219} (2008),
  894--931.

\bibitem{BK09}
A.~Borodin and J.~Kuan, \emph{{Random Surface Growth with a Wall and Plancherel
  Measures for $O(\infty)$}}, Comm. Pure Appl. Math. \textbf{63} (2010),
  831--894.

\bibitem{BP07}
A.~Borodin and S.~P\'ech\'e, \emph{{Airy Kernel with Two Sets of Parameters in
  Directed Percolation and Random Matrix Theory}}, J. Stat. Phys. \textbf{132}
  (2008), 275--290.

\bibitem{BKR09}
K.~Burdzy, W.~Kang, and K.~Ramanan, \emph{{The Skorokhod problem in a
  time-dependent interval}}, Stoch. Processes Appl. \textbf{119} (2009),
  428--452.

\bibitem{EKLP92}
N.~Elkies, G.~Kuperbert, M.~Larsen, and J.~Propp, \emph{{Alternating-Sign
  Matrices and Domino Tilings I and II}}, J. Algebr. Comb. \textbf{1} (1992),
  111--132.

\bibitem{FF10}
P.L. Ferrari and R.~Frings, \emph{{On the Partial Connection Between Random
  Matrices and Interacting Particle Systems}}, J. Stat. Phys. \textbf{141}
  (2010), 613--637.

\bibitem{FN08b}
P.~J. Forrester and T.~Nagao, \emph{{Determinantal Correlations for Classical
  Projection Processes}}, J. Stat. Mech. (2011), P08011.

\bibitem{FN08}
P.~J. Forrester and E.~Nordenstam, \emph{{The Anti-Symmetric GUE Minor
  Process}}, Mosc. Math. J. \textbf{9} (2009), 749–--774.

\bibitem{GS12}
V.~Gorin and M.~Shkolnikov, \emph{{Limits of multilevel TASEP and similar
  processes}}, To appear in Ann. Inst. H. Poincar\'e; arXiv:1206.3817 (2012).

\bibitem{HW87}
J.M. Harrison and R.J. Williams, \emph{{Multidimensional reflected Brownian
  motions having exponential stationary distributions}}, Ann. Probab. (1987),
  115--137.

\bibitem{IK10}
T.~Ichiba and I.~Karatzas, \emph{{On collisions of Brownian particles}}, Ann.
  Appl. Prob. (2010), 951--977.

\bibitem{JN06}
K.~Johansson and E.~Nordenstam, \emph{{Eigenvalues of GUE minors}}, Electron.
  J. Probab. \textbf{11} (2006), 1342--1371.

\bibitem{KPS12}
I.~Karatzas, S.~Pal, and M.~Shkolnikov, \emph{{Systems of Brownian particles
  with asymmetric collisions}}, arXiv:1210.0259 (2012).

\bibitem{Met11}
A.~Metcalfe, \emph{{Universality properties of Gelfand-Tsetlin patterns}},
  Probab. Theory Relat. Fields (2011), online first.

\bibitem{Nor08}
E.~Nordenstam, \emph{{On the Shuffling Algorithm for Domino Tilings}},
  Electron.~J.~Probab. \textbf{15} (2010), 75--95.

\bibitem{NY11}
E.~Nordenstam and B.~Young, \emph{{Domino Shuffling on Novak Half-Hexagons and
  Aztec Half-Diamonds}}, Electr. J. Comb. \textbf{18} (2011).

\bibitem{OR06}
A.~Okounkov and N.~Reshetikhin, \emph{{The Birth of a Random Matrix}}, Mosc.
  Math. J. \textbf{6} (2006), 553--–566.

\bibitem{PP08}
S.~Pal and J.~Pitman, \emph{{One-dimensional Brownian particle systems with
  rank-dependent drifts}}, Ann. Appl. Prob. (2008), 2179--2207.

\bibitem{Sko61}
A.V. Skorokhod, \emph{Stochastic equations for diffusions in a bounded region},
  Theory Probab. Appl. (1961), 264--274.

\bibitem{WW84}
S.R.S. Varadhan and R.J. Williams, \emph{{Brownian motion in a wedge with
  oblique reflection}}, Comm. Pure Appl. Math. (1984), 405--443.

\bibitem{War07}
J.~Warren, \emph{{Dyson's Brownian motions, intertwining and interlacing}},
  Electron. J. Probab. \textbf{12} (2007), 573--590.

\bibitem{WW08}
J.~Warren and P.~Windridge, \emph{{Some Examples of Dynamics for Gelfand
  Tsetlin Patterns}}, Electron. J. Probab. \textbf{14} (2009), 1745--1769.

\end{thebibliography}

\end{document}